\documentclass{article}
\usepackage[margin=2.5cm]{geometry}

\usepackage[utf8]{inputenc}
\usepackage[hidelinks]{hyperref}
\usepackage{graphicx}
\usepackage{amsmath}
\usepackage{boxedminipage}
\usepackage{amsthm}
\usepackage{amsfonts}
\usepackage{bbm}
\usepackage{algorithm}
\usepackage{algpseudocode}
\usepackage{mathtools}

\usepackage{soul}
\usepackage{cancel}
\usepackage{xspace}
\usepackage[most]{tcolorbox}
\usepackage{xcolor}
\usepackage{booktabs}
\usepackage{pifont}
\usepackage{url}
\usepackage{amssymb}
\usepackage{svg}
\usepackage{cryptocode}
\usepackage{subfig}
\usepackage{circledsteps}
\newtheorem{theorem}{Theorem}
\newtheorem{lemma}{Lemma}

\newcommand{\GAT}[0]{\ensuremath{\mathsf{GAT}}}
\newcommand{\GST}[0]{\ensuremath{\mathsf{GST}}}

\newcommand{\genesis}[0]{\ensuremath{\Lambda_{\mathsf{g}}}}
\newcommand{\sign}[2]{\ensuremath{{#1}_{#2}}}
\newcommand{\msg}[1]{\ensuremath{\langle #1 \rangle}}
\newcommand{\msgnamge}[1]{\ensuremath{\mathsf{#1}}}
\newcommand{\logmsgname}[0]{\ensuremath{\msgnamge{LOG}}}
\newcommand{\logmsg}[1]{\ensuremath{\msg{\logmsgname, #1}}}

\definecolor{pink}{rgb}{1,0.75,0.8} 

\begin{document}

\title{TOB-SVD: Total-Order Broadcast with Single-Vote Decisions in the Sleepy Model}

\author{Francesco D'Amato\\
  Ethereum Foundation\\
  \url{francesco.damato@ethereum.org}
  \and
  Roberto Saltini\\
  Ethereum Foundation\\
  \url{roberto.saltini@ethereum.org}
  \and
  Thanh-Hai Tran\\
  Independent Researcher\\
  \url{thanhhai1302@gmail.com}
  \and
  Luca Zanolini\\
  Ethereum Foundation\\
  \url{luca.zanolini@ethereum.org}
}
\date{}


\maketitle
\thispagestyle{plain}
\pagestyle{plain}

\begin{abstract}
Over the past years, distributed consensus research has expanded its focus to address challenges in large-scale, permissionless systems, such as blockchains. This shift reflects the need to accommodate dynamic participation, in contrast to the traditional model of a static set of continuously online validators. Works like Bitcoin and the sleepy model have laid the groundwork for this evolving framework.

Notable contributions by Momose and Ren (CCS 2022), along with subsequent research, have introduced Total-Order Broadcast protocols that leverage Graded Agreement primitives and support dynamic participation. However, these approaches often require multiple phases of voting per decision, which can create a bottleneck for real-world, large-scale systems.

To address this, our paper introduces TOB-SVD, a novel Total-Order Broadcast protocol in the sleepy model that is resilient to up to 1/2 Byzantine validators.
TOB-SVD requires only a single phase of voting per decision in the best case and achieves lower expected latency compared to existing approaches offering the same optimal adversarial resilience.
This work paves the way for more practical Total-Order Broadcast protocols that can be implemented in real-world systems involving a large number of validators with fluctuating participation over time.
\end{abstract}

\section{Introduction}
\label{sec:intro}
Distributed consensus research has expanded in recent years to address the challenges posed by large-scale, permissionless systems, such as blockchains. In contrast to traditional methods for solving consensus, which assume validators to be continuously online and actively contributing to the protocol, these new consensus protocols aim at accommodating dynamicity in the participation level among network validators~\cite{sleepy}. We refer to consensus (or, in our context, Total-Order Broadcast) protocols supporting dynamic participation as \emph{dynamically available}.

Momose and Ren's research~\cite{mr} laid the foundation for dynamically available Total-Order Broadcast (TOB) protocols with deterministic Safety, sparking a series of subsequent research~\cite{mmr, mmr2, gl, DBLP:conf/podc/DAmatoLZ24}. Notably, the protocols stemming from these works share a common structural theme: they all leverage a \emph{Graded Agreement}~(GA) primitive\footnote{In Graded Agreement, a grade is associated with any decision where the grade intuitively indicates how strong the agreement is.}~\cite{DBLP:conf/opodis/AttiyaW23}, albeit with different implementations and properties. For example, the Total-Order Broadcast protocol proposed by Momose and Ren~\cite{mr} (MR) utilizes a Graded Agreement protocol resilient to 1/2 Byzantine validators. However, their TOB protocol is rather impractical due to high latency. In contrast, the subsequent work by Malkhi, Momose, and Ren~\cite{mmr} (MMR) proposes two protocols (1/3MMR and 1/4MMR) which improve latency but lower the adversarial tolerance. A later refinement by the same authors~\cite{mmr2} (MMR2) reverts to tolerating minority corruption, while maintaining comparable latency. A concurrent and independent work~\cite{gl} (GL) also uses a Graded Agreement primitive (called Commit-Adopt in their work).\footnote{To be precise, GL employs a primitive called Commit-Adopt, which slightly differs from the Graded Agreement (GA) primitive~\cite{DBLP:conf/opodis/AttiyaW23}. For simplicity, we will refer to it as GA in this work, although it is important to note that there are differences between the two primitives.}

The three currently existing dynamically available protocols with optimal adversarial resilience and latency constant in the security parameter, i.e.,~\cite{mr, mmr2, gl}, all share the drawback of requiring multiple {phases} of voting\footnote{A voting phase is a point in time when every honest validator performs computations based on the messages they have received and subsequently sends a \emph{new} message. Here, ``new" refers to a message that has never been sent before, with its content determined by these computations.} for each {new block} (i.e., decision).
In practice, these protocols all operate in \emph{views}, and within each view a \emph{block} is proposed and a decision is made. To do so, two~\cite{gl}, three~\cite{mmr2}, and five~\cite{mr} instances of Graded Agreement are invoked within a view, each including four, three, and two {phases} of voting, respectively. This poses a challenge to their practical application in real-world systems, particularly in scenarios involving a multitude of simultaneous validators, such as those found in most of existing blockchains. 
{In fact, a key factor for the practical application in such systems is \emph{reducing the number of voting phases {per new block} as much as possible}.
This is because of various reasons.
First, }due to the large number of validators involved, voting phases are computationally intensive which results in an extra computation delay added on top of the network latency.
{Additionally, these protocols often require a signature aggregation process where messages are first sent to aggregators who then distribute the aggregated signatures, causing voting phases to require double the normal network latency.}
For instance, the Ethereum's consensus protocol divides its \emph{slot} -- the unit of time in which a new block is proposed and voted -- in three phases taking up $\Delta = 4s$ each. One $\Delta$ is for \emph{block} propagation and two $\Delta$ are for \emph{vote} propagation, one of which is reserved to propagation and aggregation of votes, and one to propagation of aggregated votes. Effectively, a voting phase in Ethereum takes $2\Delta$\cite{specs}. 

We extend this line of work by improving the practicality of dynamically available protocols with deterministic Safety. Specifically, our main contributions include:
\begin{enumerate}
    \item \textbf{Introduction of a dynamically available Total-Order Broadcast protocol} that tolerates up to 1/2 Byzantine validators.
    \item \textbf{Latency comparable to the state-of-the-art MMR2 protocol}~\cite{mmr2} (Table~\ref{tab:comparison}), with:
    \begin{itemize}
        \item Slightly better performance in the {expected case}.
        \item Slightly worse performance in the best case.
    \end{itemize}
    \item {\textbf{Better suited for large-scale permissionless systems} by lowering the number of voting phases {per new block} to one {in the best case, and to two in the expected case}.}
\end{enumerate}

Expanding on point 3, a recent research direction demonstrates how combining a dynamically available total-order broadcast protocol with a \emph{finality gadget} -- a partially synchronous total-order broadcast protocol -- can produce a secure \emph{ebb-and-flow} protocol~\cite{DBLP:conf/sp/NeuTT21}. Specifically, an ebb-and-flow protocol consists of two components: (1) a dynamically available total-order broadcast protocol that guarantees safety and liveness under synchrony ($\GST = 0$) at all times, and (2) a finalizing protocol that ensures safety at all times and liveness after $\max(\GST, \GAT)$, where $\GAT$ represents a global awake time after which all honest validators come online. 

Practical research has further explored~\cite{DBLP:conf/esorics/DAmatoZ23}, within the context of the Ethereum consensus protocol, how to combine a probabilistically safe dynamically available total-order broadcast protocol~\cite{rlmd} and a finality gadget in order to achieve \emph{single-slot finality}~\cite{buterin2024single}. Building on this, we \emph{strongly} believe that similar results can be achieved by replacing their dynamically available protocol with the protocol presented in this work, potentially reducing message sizes.

The remainder of this work is structured as follows. We present and compare related work in Section~\ref{sec:related-work}. Section~\ref{sec:model} presents the system model and essential definitions. Importantly, we revisit the sleepy model~\cite{sleepy}, expanding upon the notation first introduced by Malkhi, Momose, and Ren~\cite{mmr2}. Section~\ref{sec:background} recalls a foundational result that serves as the starting point for our protocol. Section~\ref{sec:tob-svd} details our contributions. Specifically, Section~\ref{sec:2valGA} and Section~\ref{sec:3valGA} introduce two Graded Agreement protocols: the first, a two-grade GA protocol, lays the groundwork for the more advanced three-grade GA protocol. Section~\ref{sec:our-tob} then presents TOB-SVD, our dynamically available Total-Order Broadcast protocol. We provide a comprehensive analysis of both Graded Agreement protocols and TOB-SVD in Section~\ref{sec:analysis}. Finally, Section~\ref{sec:conclusions} concludes the paper.

\section{Related Work}
\label{sec:related-work}

\begin{table*}
    \centering
    \scriptsize
    \renewcommand{\arraystretch}{1.5}
    \begin{tabular}{|c|c|c|c|c|c|c|}
        \hline
        & \rotatebox{90}{TOB-SVD} & \rotatebox{90}{MR~\cite{mr}} & \rotatebox{90}{MMR2~\cite{mmr2}} & \rotatebox{90}{GL~\cite{gl}} & \rotatebox{90}{1/3MMR~\cite{mmr}} & \rotatebox{90}{1/4MMR~\cite{mmr}} \\ \hline  
        Adversarial resilience & $1/2$ & $1/2$ &  $1/2$ & $1/2$ & $1/3$ & $1/4$ \\ \hline 
        Best-case latency & $6\Delta$ &  $16\Delta$ & $4\Delta$ & $10\Delta$ & $3\Delta$ & $2\Delta$ \\ \hline 
        Expected latency & $10\Delta$ & $32\Delta$ & $14\Delta$ & $20\Delta$ & $6\Delta$ & $4\Delta$ \\ \hline 
        Transaction expected latency & $12\Delta$ & $50.5\Delta$ & $19\Delta$ & $25\Delta$ & $7.5\Delta$ & $5\Delta$ \\ \hline 
        Voting phases per {new} block in the best case & $1$ & $10$ & $3$ & $5$  &  $2$ & $1$ \\ \hline 
        Voting phases per {new} block in the expected case & $2$ & $20$ & $12$ & $10$  &  $4$ & $2$ \\ \hline 
        {Communication complexity} & $O(Ln^3)$ & $O(Ln^3)$ & $O(Ln^3)$ & $O(Ln^3)$ & $O(Ln^2)$ & $O(Ln^2)$\\ \hline 
    \end{tabular}
    \vspace{0.3cm}
    \caption{\textnormal{Comparison of dynamically available TOB protocols with deterministic Safety, where $L$ represents the block size and $n$ refers to the number of validators.}}
    \label{tab:comparison}
\end{table*}

To provide a clearer motivation for our results, we first compare the outcomes of our work with those of other protocols. 
To facilitate this comparison, we introduce key metrics to help evaluate the performance and characteristics of protocols.

{In the following, we assume that upon submission, transactions are immediately added to a transaction pool} from which validators can retrieve and validate them using a specified validity predicate\footnote{For simplicity, we assume that the validity of a transaction is evaluated independently of any other transaction already included in the log. A transaction is valid according to a global, efficiently computable predicate $P$, known to all validators. The specific details of this predicate are omitted.} before batching them into blocks.
{These blocks are then appended to an existing sequence of other blocks, or a log (Section~\ref{sec:primitives}), and such a resulting log is proposed for decision by the consensus protocol.
We assume that honest validators batch into any proposed block any valid transaction included in the transaction pool that is not already included in the log that the proposed block is appended to.}

The \emph{confirmation time} of a transaction is defined as the time between its submission (from a user) and the decision regarding the log containing that transaction. 

We refer to \emph{best-case latency} as the minimum possible confirmation time of a transaction, where the minimum is taken over all possible submission times. In practice, this corresponds to the shortest time between a proposal and its decision.

{We refer to \emph{expected latency} as the \emph{expected} confirmation time of a transaction submitted right before the next proposal.}

{Finally, we refer to \emph{transaction expected latency} as the expected confirmation time of a transaction submitted at a \emph{random time}, which is equivalent to the sum of half of the time between consecutive proposals and expected latency.}

In Table~\ref{tab:comparison} we compare all existing dynamically available Total-Order Broadcast protocols that ensure deterministic Safety, where $L$ represents the block size and $n$ refers to the number of validators.

Analyzing best-case latency, the MR protocol~\cite{mr} proves impractical due to its high latency of approximately $16\Delta$\footnote{where $\Delta$ is the network delay bound}. In contrast, the 1/3MMR and 1/4MMR protocols~\cite{mmr} reduce best-case latency to $3\Delta$ and $2\Delta$, respectively, but at the cost of lowering adversarial tolerance to 1/3 and 1/4, compared to the 1/2 tolerance of the MR protocol. Maintaining the 1/2 adversarial tolerance, the MMR2 protocol~\cite{mmr2} achieves a best-case latency of $4\Delta$, while the GL protocol~\cite{gl} experiences a higher best-case latency of $10\Delta$. Our protocol slightly increases best-case latency to $6\Delta$, while preserving the 1/2 adversarial tolerance.

Our protocol's expected latency is $10\Delta$, which improves on prior works: $14\Delta$, $20\Delta$, and $32\Delta$ for the MMR2, GL, and MR protocols, respectively. Specifically, under our assumption regarding leader election (Section~\ref{sec:common-notion}), {our protocol utilizes a Verifiable Random Function (VRF)-based leader election mechanism. Each transaction is expected to be batched by an honest proposer within two views, resulting in an expected confirmation time of $10\Delta$. 
In contrast, the 1/3MMR and 1/4MMR protocols reduce their adversary tolerance to achieve comparable latency improvements.

Regarding transaction expected latency, our protocol achieves an expected latency of $12\Delta$, which contrasts with MR at $50.5\Delta$, MMR2 at $19\Delta$, and GL at $25\Delta$. We believe this metric is particularly relevant in practice, as users typically submit transactions at random times -- whenever they wish to transact -- without waiting for block delivery to submit them immediately afterward.

{
    As introduced earlier, existing TOB protocols with deterministic Safety operate in {views}, where a {block} is proposed and a decision is made within each view. 
    We analyze the number of voting phases per {new} block of all protocols in Table~\ref{tab:comparison} in both the best case and the expected case.

    In the best case, our protocol requires only one voting phase per {new} block, which improves on prior works: 10, 5, 3, and 2 voting phases per {new} block for the MR, GL, MMR2, and 1/3MMR protocols, respectively.
    Like our protocol, the 1/4MMR protocols only needs a single voting phase.

    Analyzing the expected case, the GL protocol uses 10 voting phases.
    The MMR2 protocol requires 4 GA instances, each involving 3 voting phases, resulting in a total of 12 voting phases per {new} block, while the MR protocol requires 10 GA instances with 2 voting phases each, totaling 20 voting phases per {new} block. 
    Similarly, the 1/3MMR and 1/4MMR protocols need 4 and 2 voting phases, respectively.
    Finally, our protocol needs only 2 voting phases per {new} block.
}

Our protocol achieves the same expected communication complexity as MMR2, GL, and MR which is $\mathit{O}(Ln^3)$, where $L$ represents the block size.
{Both 1/3MMR and 1/4MMR have a communication complexity of only $O(Ln^2)$ as they do not forward received messages, while all others do. However, as highlighted above, this goes to the detriment of adversarial resiliency.}

Finally, our protocol requires a \emph{stabilization period} of $2\Delta$ as a fundamental assumption for its security. Other protocols, such as 1/3MMR, 1/4MMR, and MMR2, do not require such a period. However, in practice, both these protocols and ours necessitate of a recovery mechanism to allow validators that wake up to recover messages that were sent to them while asleep, since assuming that messages are buffered and delivered immediately is not very practical.

Such recovery protocol typically requires that, upon waking up, a validator sends a \msgnamge{RECOVERY} message to other validators.
These validators then send back any messages that the newly awakened validator may have missed while asleep and that could impact future decisions.
The validator that wakes up is required to remain awake until it receives responses to the \msgnamge{RECOVERY} messages it has sent out.
Only then it is considered awake in the context of the adversary model constraints.
Such a period is, in practice, at least $2\Delta$.

Therefore, in practice, we share with other related works the requirement for a period of stable honest validator participation. For our protocol, this requirement is stronger than that of 1/3MMR, 1/4MMR, and MMR2, as we need stable participation during both the recovery period and the stabilization period. This means that, overall, our protocol's stabilization assumption is $2\Delta$ longer than that of other protocols. However, note that, in practice, when a validator wakes up, it may need to download a significant amount of data to catch up, potentially requiring more than just \(2\Delta\). For example, if the recovery procedure takes several hours, assuming an additional stabilization period of \(2\Delta\) should not significantly impact the overall protocol.

\section{Model and Definitions}
\label{sec:model}

\subsection{System model} 

\textbf{Validators.} We consider a system of~$n$ \emph{validators} $\mathcal{V}=\{v_1,\dots,v_n\}$ in a message-passing system with an underlying \emph{synchronous} network with delay bound~$\Delta > 0$. We assume that validators have synchronized clocks. Each validator is uniquely identified by a cryptographic identity, and their public keys are common knowledge.
{We use the notation $\sign{\msg{m}}{i}$ to indicate that a message $\msg{m}$ is signed by validator $v_i$.}
A protocol for~$\mathcal{V}$ consists of a collection of programs with instructions for all validators.

\textbf{Honest and Byzantine validators.} A validator that follows its protocol during an execution is called \emph{honest}. At any time, honest validators forward any message received.
On the other hand, a faulty validator may deviate arbitrarily from its specification, e.g., when corrupted by an adversary; we refer to such validators as \emph{Byzantine}. 

\textbf{Adversary and sleepy model.}
We assume an \emph{adversary} capable of permanently corrupting honest validators throughout the execution in a \emph{mildly adaptive} manner. To be more specific, if the adversary corrupts an honest validator~$v_i$ at time~$t$, then~$v_i$ becomes Byzantine only at time~$t+\Delta$. In other words, we assume that the adversary does not have the capability to corrupt honest validators immediately. Instead, there is a delay, represented by \( \Delta \), before it can do so. This delay is essential for the effective operation of a VRF-based leader election, as detailed in Section~\ref{sec:common-notion}.

The adversary can also \emph{fully adaptively} either \emph{put validators to sleep}, preventing them from participating in the protocol, or \emph{wake them up}, enabling them to participate in the protocol. Validators subjected to the sleeping state are referred to as \emph{asleep}, while those actively participating in the protocol are designated \emph{awake}. We assume that, upon waking up, validators immediately receive all messages they should have received while asleep.\footnote{This assumption, while not practical for real-world systems, serves as a theoretical framework for analyzing the protocol, consistent with prior works~\cite{sleepy, mr, mmr, mmr2}. As briefly discussed in Section~\ref{sec:related-work}, incorporating a recovery procedure into the protocol could eliminate the need for this assumption, aligning it with practical considerations. However, the design and integration of such a recovery procedure in our protocol fall outside the scope of this work.} 
In other words, we are considering (a variant of) the sleepy model~\cite{sleepy}, which we now fully specify.

In the original formulation of the sleepy model~\cite{sleepy}, {like in our variant,} the adversary can at any time determine the state -- either \emph{awake} or \emph{asleep} -- of any honest validator. Validators in the asleep state do not participate in the protocol; messages intended for them during this state are enqueued, only to be delivered in the subsequent time step when the validator is reawakened. 

Additionally, Byzantine validators remain always awake. The sleepy model, in fact, does not allow for fluctuating participation among Byzantine validators. The reason for this assumption is due to a problem called \emph{costless simulation}~\cite{DBLP:conf/fc/DebKT21}: Byzantine validators can exploit both past and future activity to compromise consensus. When awaken, a Byzantine validator might mimic continuous past engagement, generating messages to retroactively alter outcomes (\emph{backward simulation}). Additionally, these validators can conduct \emph{forward simulations} by sharing their secret keys with allied Byzantine validators creating an illusion of persistent activity.

Variants of the sleepy model, and subsequent protocols working in such variants, have been devised~\cite{mmr, mmr2, goldfish, rlmd, DBLP:conf/podc/DAmatoLZ24}, allowing the set of Byzantine validators to grow. These models can, at a minimum, address backward simulation. This particular approach, which is the one that we adopt in this paper, is referred to as the \emph{growing adversary} model. Malkhi, Momose, and Ren~\cite{mmr2} provide a formalization of this model using specific parameters. We adopt (a slight variant of) this formalization for the remaining part of our work.

To be precise, for any time $t \geq 0$, let $H_t$ be the set of awake honest validators at time $t$, and let $B_t$ be the set of Byzantine validators at time $t$\footnote{If $t<0$ we define $H_t:=\mathcal{V}$ and $B_t:=\emptyset$.}. We assume that $t_1 \leq t_2 \implies B_{t_1}\subseteq B_{t_2}$, i.e., we assume $B_t$ to be monotonically \emph{non-decreasing}\footnote{The reason for this assumption is due to forward simulation.}.
For the interval \([t_1, t_2]\), we define $H_{t_1, t_2} = (\bigcap_{t \in [t_1, t_2]} H_t)$ to be the set containing the honest validators that were awake between~\( t_1 \) and~\( t_2 \).

Let $T_b$ and $T_s$ be non-negative integers representing two specific time constants.
The parameter $T_b$ represents the duration for which backward simulation~\cite{mmr2} is considered. In other words, Byzantine validators are counted for an extra time $T_b$ backward. In other words, $T_b$ captures the requirement that not too many honest validators become corrupted within \( T_b \) after performing a specific action, such as voting, in order to prevent them from sending contradictory messages shortly afterward. Our model is additionally augmented with the stabilization period $T_s$ in order to encapsulate a \emph{stable participation} requirement. This is emphasized by focusing on honest validators who remain awake throughout the time span~\( T_s \). 

Then, we define the \emph{active validators at time $t$} to be the set $H_{t - T_s, t} \cup B_{t + T_b}$.
Intuitively, this corresponds to the smallest set of validators that might\footnote{By ``might'' we mean that there exists an execution where this happens.} send a message during a GA instance\footnote{This concept will be formally introduced later.} starting at time $t$ and lasting $T_b$.
More generally, this set corresponds to the smallest set of validators that might influence outputs/decisions dependent on messages sent by honest validators at time $t$.

Furthermore, let~$\rho \le \frac{1}{2}$ denote a predetermined failure ratio of Byzantine to active validators.

A system is compliant with the $(T_b, T_s, \rho)$-\emph{sleepy model}
if and only if, for every time $t \geq 0$, the following condition

is satisfied:
\begin{equation}
\label{eq:sleepy}
 |B_{t + T_b}| < \rho \cdot |H_{t - T_s, t} \cup B_{t + T_b}| 
\end{equation}

Finally, as long as a validator remains honest, the adversary cannot forge its signatures.

\subsection{Graded Agreement and Total-Order Broadcast}
\label{sec:primitives}

\textbf{Logs.} We define a \emph{log} as a finite sequence of \emph{blocks}~$b_i$, represented as \(\Lambda = [b_1,b_2,\dots,b_k]\). Here, a block represents a batch of \emph{transactions} and it contains a reference to another block. In this work, we assume that there exists an external pool of transactions. Honest validators retrieve transactions from this pool and validate them using a specified validity predicate before batching them into blocks. 

Given two logs, \(\Lambda\) and \(\Lambda'\), the notation \(\Lambda \preceq \Lambda'\) indicates that \(\Lambda\) is a prefix of \(\Lambda'\).
Two logs are \emph{compatible} if one acts as a prefix for the other. Conversely, if neither log is a prefix of the other, they are \emph{conflicting}. 
Finally, we say that a log $\Lambda'$ is an \emph{extension} of $\Lambda$ if {and only if} $\Lambda$ is a prefix of $\Lambda'$.
{We assume that any log is an extension of a log $\genesis$ known to any validator.\footnote{In blockchain protocols, $\genesis$ typically consists of a log of length $1$.}}

\textbf{Graded Agreement.} We define a generic \emph{Graded Agreement} (GA) primitive, with $k \ge 1$ grades. In such a primitive, each validator has an input {log}~$\Lambda$, and validators can output \emph{logs} with grade~$0 \leq g < k$, which we denote with the pair~$(\Lambda, g)$.

We refer to the phase where a log~$\Lambda$ is inputted into the primitive as the \emph{input phase}, and the phase during which the output is retrieved from the Graded Agreement primitive as the \emph{output phase}. For this reason, we often refer to a log input by a validator~$v_i$ into the GA as \emph{the input of validator~$v_i$}. Such a primitive can have different output phases (potentially up to~$k$) but only one input phase. 

A validator~$v_i$ that is awake in the output phase for grade~$g$ \emph{may attempt} to output a log with grade~$g$, i.e., to run an output procedure, possibly resulting in outputting some log, potentially more than one, with grade~$g$. If validator~$v_i$ attempts to do so, we say that~$v_i$ \emph{participates in the output phase for grade~$g$}. The criteria which validators use to decide whether to participate in output phases are specific to each implementation of the Graded Agreement primitive, with the caveat that \emph{honest validators that are always awake participate in every output phase}, so that outputs can at a minimum be guaranteed when honest participation is \emph{stable}. 

Keeping in mind that honest validators can output multiple logs with the same grade, we require the Graded Agreement primitive to satisfy the following properties.

\begin{enumerate}
    \item \textbf{Consistency:} 
    If an honest validator outputs $(\Lambda, g)$ for $g > 0$, then no honest validator outputs $(\Lambda', g)$ with $\Lambda'$ conflicting with $\Lambda$.

    \item \textbf{Graded Delivery:} 
    If an honest validator outputs $(\Lambda, g)$ for $g > 0$, any honest validator that participates in the output phase for grade $g-1$ outputs $(\Lambda, g-1)$.

    \item \textbf{Validity:} 
    {If each honest validator awake at time $0$ inputs a potentially different extension of a log $\Lambda$,}
    then all honest validators participating in the output phase for a grade $g$ output $(\Lambda, g)$.

    \item \textbf{Integrity:} 
    If no honest validator inputs an extension of $\Lambda$, then no honest validator outputs $(\Lambda, *)$.

    \item \textbf{Uniqueness:} 
    An honest validator does not output $(\Lambda, g)$ and $(\Lambda', g)$ for $\Lambda$ conflicting with $\Lambda'$.
\end{enumerate}

Note that, for $g > 0$, Consistency already implies Uniqueness. It is only a separate property for outputs of grade $0$.

\textbf{Total-Order Broadcast.}
A  Total-Order Broadcast (TOB) protocol ensures that all the honest validators \emph{deliver} the same log~$\Lambda$. 

A protocol for (Byzantine) Total-Order Broadcast satisfies the following properties.

\begin{enumerate}
    \item \textbf{Safety:} 
    If two honest validators deliver logs $\Lambda_1$ and $\Lambda_2$, then $\Lambda_1$ and $\Lambda_2$ are compatible.

    \item \textbf{Liveness:} For every \emph{valid} transaction $\mathit{tx}$ {in the pool of transactions}, there exists a time $t$ such that all honest validators awake for sufficiently long\footnote{The duration ``sufficiently long" varies depending on the protocol.} after $t$ deliver a log $\Lambda$ that includes  (a block that includes) transaction $\mathit{tx}$.
\end{enumerate}

      Given a choice of parameters $(T_b, T_s, \rho)$ as previously introduced when presenting the sleepy model, we say that a Total-Order Broadcast protocol is \emph{dynamically available} if it is a Total-Order Broadcast in the $(T_b, T_s, \rho)$-sleepy model.

\subsection{Common notions}
\label{sec:common-notion}
\textbf{Messages.} Our protocol defines only one type of message, the \logmsgname{} message $\logmsg{\Lambda}$ where $\Lambda$ is a log.
Informally, we say that validator $v_i$ sends log $\Lambda$ to mean that $v_i$ sends a message $\sign{\logmsg{\Lambda}}{i}$.
Similarly, we say that a validator $v_j$ receives a log from validator $v_i$  to mean that validator $v_j$ receives the message $\sign{\logmsg{\Lambda}}{i}$.

\textbf{Equivocations.} We refer to multiple different {\logmsgname{} messages} from the same validator as \emph{equivocations}, and to any pair of such {\logmsgname{} messages} as \emph{equivocation evidence} for its sender. Honest validators only ever accept and forward up to two {\logmsgname{} messages} per validator. Validators that are known to have sent an equivocation are called \emph{equivocators}.

\textbf{Leader election.}
Our Total-Order Broadcast protocols proceeds in \emph{views}, and employs a VRF-based leader election primitive~\cite{sleepy}. Each validator has an associated VRF value for each view. Whenever a \emph{proposal} has to be made to extend the current log, validators broadcast one together with their VRF value for the current view, and priority is given to proposals with a higher VRF value. Since this is a standard tool for leader election in dynamically available protocols~\cite{mr,mmr, mmr2, goldfish, gl}, and not the focus of this work, we use the VRF primitive informally. 

As mentioned at the beginning of this section, such a leader election requires us to consider a mildly adaptive corruption model, where the adversary has a delay of time $\Delta$ between scheduling a corruption and executing it. This appears to be necessary also in other protocols which use this strategy~\cite{DBLP:conf/podc/DAmatoLZ24,goldfish,rlmd}. To see why, consider the usual leader election where VRF values are broadcast at time $t$ and a leader is chosen at time $t+\Delta$ based on the highest VRF value so far observed. Between time $t$ and $t+\Delta$, an adaptive adversary can observe the highest VRF value and corrupt its sender, then have it deliver an equivocating proposal only to a subset of the honest validators. This way, some subset of the honest validators only knows of a single proposal with the highest VRF value, and some other subset knows of two such proposals. We cannot then ensure that all honest validators \emph{vote}\footnote{The terms ``propose" and ``vote" are used here informally to appeal to the reader's intuition. These terms will be formally defined in the following section.} for the same proposal, which is typically required by {the} Liveness arguments.

For the following, we define a \emph{good leader for view~$v$ starting at time $t_v$} to be a validator in~$H_{t_v} \setminus B_{t_v + \Delta}$ holding the highest VRF value for view $v$ among validators $H_{t_v} \cup B_{t_v + \Delta}$, i.e., among all validators from which a proposal for view $v$ might be received by time $t_v + \Delta$.
Note that a good leader always \emph{proposes something}.

\textbf{Validator state.} At all times, an honest validator keeps only two local variables, $V$ and $E$. First, $V$ associates to a validator~$v_i$ the log {$V(i) = \sign{\logmsg{\Lambda}}{i}$} if it has received an unique message {\sign{\logmsg{\Lambda}}{i} (from validator $v_i$)}, or $V(i) = \bot$ if either none or at least two messages $\logmsg{\Lambda}$ and $\logmsg{\Lambda'}$ with $\Lambda \neq \Lambda'$ have been received from $v_i$. In other words, $V$ keeps track of non-equivocating {\logmsgname{} messages}. We write~$v_i \in V$ if {and only if} $V(i) \not = \bot$, and write~$V_{\Lambda}$ for the {set} containing all extensions of~$\Lambda$ recorded in $V$, paired with their sender, i.e., {$V_{\Lambda} = \{(\Lambda', v_i): v_i \in V, V(i) = \sign{\logmsg{\Lambda'}}{i}, \Lambda'\succeq \Lambda\}$}.
{We say that a validator $v_i$ is in $V_\Lambda$ to mean that there exists a pair $(*, v_i) \in V_\Lambda$.
Similarly, we say that a log $\Lambda'$ is in $V_\Lambda$ to mean that there exists a pair $(*, \Lambda') \in V_\Lambda$.}

Moreover, $E$ contains a record of equivocators and equivocation evidence, i.e., $E(i) = \bot$ if $v_i$ is not known to have equivocated, and otherwise {$E(i) = (\sign{\logmsg{\Lambda}}{i}, \sign{\logmsg{\Lambda'}}{i})$, where $(\sign{\logmsg{\Lambda}}{i}, \sign{\logmsg{\Lambda'}}{i})$} is equivocation evidence for validator~$v_i$. As for $V$, we write $v_i \in E$ if {and only if} $E(i) \not = \bot$. A validator can compute from $V$ and $E$ the set $S = \{v_i \in \mathcal{V}: v_i \in V \lor v_i \in E\}$ of all {the senders of \logmsgname{} messages}, i.e., of all validators from which \emph{at least one {\logmsgname{} message}} has been received.

When we want to emphasize the time $t$ at which we consider these variables, we write~$V^t$, $E^t$, and~$S^t$. If we also want to emphasize the validator~$v_i$ whose sets we consider, we write~$V^{t,i}$, $E^{t,i}$, and~$S^{t,i}$.  

\textbf{Message handling.} If message $\sign{\logmsg{\Lambda}}{i}$ is received at time $t$, we have three possibilities on how to handle it. 
If~$v_i \not \in V$, i.e., we have not received a $\logmsgname$ message from~$v_i$ yet, we add $\Lambda$ to $V$ and forward the message. 
If~$v_i \in V$ and {$V(i) = \sign{\logmsg{\Lambda'}}{i}$} for $\Lambda' \neq \Lambda$, i.e., we are first learning about an equivocation from~$v_i$, we record the equivocation by setting $V(i) = \bot$ and {$E(i) = (\sign{\logmsg{\Lambda}}{i},\sign{\logmsg{\Lambda'}}{i})$}. 
Moreover, we also forward the message, to make sure other validators also learn about the equivocation. 
Finally, $\sign{\logmsg{\Lambda}}{i}$ is ignored if~$v_i \in E$, i.e., we already know~$v_i$ is an equivocator.

\section{Background}
\label{sec:background}

Momose and Ren~\cite{mr} introduced the first \emph{quorum-based} Graded Agreement protocol working in the sleepy model~\cite{sleepy}, achieving optimal adversarial resilience. Tolerating dynamic participation necessitates doing away with absolute quorums (e.g., 1/2 of the validators), and instead defining them based on the participation level at a given point in time. However, validators can have different perceived participation levels, making agreement challenging. 
To overcome this, Momose and Ren {\cite{mr}} introduce the novel \emph{time-shifted quorum} technique which we now summarize.
Note that, to limit the notation that we need to introduce to explain such a protocol, we adapt the original protocol by Momose and Ren to work on logs rather than values {like} in the original formulation as
the key concept of time-shifted quorum that we want to illustrate here is not affected by such a change.
We let $X_\Lambda$\footnote{In the original work by Momose and Ren~\cite{mr}, this is denoted as $\mathcal{E}(b)$, where $b$ represents a value rather than a log. Similarly, the following terminology equivalences apply: $V_\Lambda$ corresponds to $\overline{\mathcal{E}}(b)$ in their work, while $S^{{2\Delta}}$ is identified with $\mathcal{M}_1$ and {$S^{{3\Delta}}$ is identified with} $\mathcal{M}_3$ in~\cite{mr}.} be the set of all validators from which we have received a message $\logmsg{\Lambda'}$ with $\Lambda' \succeq \Lambda$, regardless of whether we have also received any different {\logmsgname{} message}, i.e., regardless of whether they have equivocated. At any time, honest validators forward any message received.
Then, the protocol executed by an honest validator $v_i$ upon \emph{inputting} $\Lambda$ is as follows. 

\begin{enumerate}
    \item 
    At $t=0$, broadcast $\sign{\logmsg{\Lambda}}{i}$.

    \item 
    At $t=\Delta$, store $V^\Delta$.

    \item 
    At $t = 2\Delta$, a validator  sends a $\msgnamge{VOTE}$ message for $\Lambda$
    if the \emph{current} support level (without discarding equivocations) is greater than half of the perceived participation level, i.e, if $|X_\Lambda^{2\Delta}| > |S^{2\Delta}|/2$.

    \item 
    At $t=3\Delta$, output $(\Lambda,1)$ if  $|V_\Lambda^{\Delta}| > |S^{3\Delta}|/2$, and $(\Lambda,0)$ if the number of \msgnamge{VOTE} messages for a log~$\Lambda' \succeq \Lambda$ is greater than half of all received \msgnamge{VOTE} messages.
\end{enumerate}

The key insight is that, if an honest validator~$v_i$ outputs~$(\Lambda,1)$ by seeing $|V_\Lambda^{\Delta,i}| > |S^{3\Delta,i}|/2$, then an honest validator~$v_j$ awake at time~$2\Delta$ would have $X_\Lambda^{2\Delta,j}$ and $S^{2\Delta,j}$ such that $|X_\Lambda^{2\Delta,j}| \ge |V_\Lambda^{\Delta,i}| > |S^{3\Delta,i}|/2 \ge |S^{2\Delta,j}|/2$, and thus would vote for~$\Lambda$. The first inequality holds because any \logmsgname{} message seen by~$v_i$ at time~$\Delta$, and thus counted in $V_\Lambda^{\Delta,i}$, is forwarded and received by~$v_j$ by time~$2\Delta$, and so counted in $X_\Lambda^{2\Delta,j}$ as well. The converse applies to $|S^{2\Delta,j}|/2$ and $|S^{3\Delta,i}|/2$, i.e., $|S^{2\Delta,j}|/2$ is determined $\Delta$ time before $|S^{3\Delta,i}|/2$, justifying the last inequality. In other words, an honest validator outputting~$(\Lambda,1)$ implies that all honest validators awake at time~$2\Delta$ vote for~$\Lambda$, and thus, {given that $\rho \le \frac{1}{2}$}, also that all honest validators output~$(\Lambda,0)$.

Notably, this protocol counts \emph{all \logmsgname{} messages}, including equivocations, when determining $|X_\Lambda^{2\Delta}|$, though it does not do so when determining $|V_\Lambda^{\Delta}|$. 
This is crucial in the time-shifted quorum argument, because then all \logmsgname{} messages that count for $|V_\Lambda^\Delta|$ are guaranteed to count for $|X_\Lambda^{2\Delta}|$ as well. 

A compromise of this approach is that it prevents this GA primitive from satisfying the Uniqueness property on values with grade~$0$. The same limitation applies to the variant presented in~\cite{mmr2}. Implementing a Total-Order Broadcast based on these primitives introduces significant complexity.

\section{TOB-SVD}
\label{sec:tob-svd}

We begin this section by presenting the foundational building block of our total-order broadcast protocol. 
A detailed analysis of our results can be found in Section~\ref{sec:analysis}.

\subsection{Graded Agreement with $k=2$ grades.}
\label{sec:2valGA}
Our first Graded Agreement protocol is a variant of the graded agreement of MR~\cite{mr} (Section~\ref{sec:background}), designed to also satisfy Uniqueness for every grade. The protocol is {given in Figure~\ref{fig:ga-two-values}}.
It lasts $3\Delta$ time, and requires that the number of Byzantine validators is less than half the number of active validators.
Therefore, it specifically works in the $(3\Delta, 0,\frac{1}{2})$-sleepy model.

\begin{figure}[h!]
    \centering
    \begin{boxedminipage}[t]{\columnwidth}\small
        \textbf{Upon input} $\Lambda$, a validator~$v_i$ runs the following algorithm whenever awake. All validators awake at time~$2\Delta$ participate in the output phase for grade~$0$. A validator participates in the output phase for grade~$1$ at time~$3\Delta$ if {and only if} it was awake also at time~$\Delta$. At any time, honest validators forward any message received. Up to two different \logmsgname{} messages per sender are forwarded upon reception.

\begin{enumerate}
    \item \textbf{Input phase, $(t = 0)$:} Broadcast $\sign{\logmsg{\Lambda}}{i}$.

    \item \textbf{$(t = \Delta)$:} Store $V^\Delta$.

    \item \textbf{Output phase for grade $0$, $(t = 2\Delta)$:} 
    If $|V^{2\Delta}_{\Lambda}| > |S^{2\Delta}|/2$: Output $(\Lambda, 0)$.

    \item \textbf{Output phase for grade $1$, $(t = 3\Delta)$:} 
    If awake at time $\Delta$: If $|V^{\Delta}_{\Lambda} \cap V^{3\Delta}_{\Lambda}| > |S^{3\Delta}|/2$, Output $(\Lambda, 1)$.
\end{enumerate}
 
    \end{boxedminipage}
    \caption{Graded Agreement protocol with $k=2$ grades -- protocol for validator~$v_i$.}
    \label{fig:ga-two-values}
\end{figure}

The protocol still relies on the key ideas of the time-shifted quorum technique presented above. A notable difference is that grade~$0$ outputs in MR~\cite{mr} are computed at time~$3\Delta$ using \msgnamge{VOTE} messages, whereas our GA protocol does not have any other message other than \logmsgname{} messages, and it computes grade~$0$ outputs at time~$2\Delta$. That said, votes in their GA are cast (almost) in the same way as grade~$0$ outputs are computed in our protocol, and they are essentially used to propagate the information forward until the end of the protocol. Another difference is the use of $|V_\Lambda^{\Delta} \cap V_\Lambda^{{3\Delta}}|$ when determining outputs of grade~$1$, which is related to the treatment of equivocations. Let us initially ignore that, and pretend that we output $(\Lambda, 1)$ when $|V_{\Lambda}^\Delta| > |S^{3\Delta}|/2$.

If we assumed that no equivocation is possible, we would get that the inequalities $|V^{2\Delta, j}| \ge |V^{\Delta, i}| > |S^{3\Delta, i}|/2 \ge |S^{2\Delta, j}|/2$ hold for validators~$v_i$ and $v_j$, when~$v_i$ outputs~$(\Lambda, 1)$, analogously to $|X_\Lambda^{2\Delta,j}| \ge |V_\Lambda^{\Delta,i}| > |S^{3\Delta,i}|/2 \ge |S^{2\Delta,j}|/2$ in the previous GA. This would immediately give us that~$v_j$ outputs~$(\Lambda, 0)$, i.e., Graded Delivery. However, the inclusion \(V^{\Delta, i} \subseteq V^{2\Delta, j}\) is not guaranteed once equivocations are allowed, because validator~\(v_j\) might discard some of the logs in~\(V^{\Delta, i}\) between time~\(\Delta\) and time~\(2\Delta\), if such logs turned out to be equivocations. 

To ensure that the supporting \logmsgname{} messages used when attempting to output a log $\Lambda$ with grade~\(0\) are more than when doing so with grade~\(1\), we use \(V^{\Delta}_\Lambda \cap V^{3\Delta}_\Lambda\) instead of just \(V^{\Delta}_\Lambda\). Crucially, any validator~$v_k$ which is seen as an equivocator by validator~$v_j$ at time $2\Delta$ will also be seen as an equivocator by validator~$v_i$ at time $3\Delta$, since~$v_j$ would forward the equivocating \logmsgname{} messages of~$v_k$ upon receiving them. In particular, {logs sent by}~$v_k$ will \emph{not} be contained in~$V^{3\Delta, i}$, and thus also not in $V^{\Delta, i}_\Lambda \cap V^{3\Delta, i}_\Lambda$. This ensures that $V^{\Delta, i}_\Lambda \cap V^{3\Delta, i}_\Lambda \subseteq V^{2\Delta, j}_\Lambda$, since any log in $V^{\Delta, i} \setminus V^{2\Delta, j}$ is an equivocation, which {implies that it is} absent from $V^{3\Delta, i}$. In other words, we separate the initial determination of possible supporting  \logmsgname{} messages from the final determination of which  \logmsgname{} messages should be treated as equivocations and removed, and we \emph{apply the time-shifted quorum technique to the set of equivocators} as well: like the perceived participation level, the set of equivocators \emph{increases} when going from the output phase for grade~$0$ to the output phase for grade~$1$.

\subsection{Graded Agreement with $k=3$ grades.}
\label{sec:3valGA}

Building upon our Graded Agreement protocol with $k=2$ grades, we extend it to a Graded Agreement with $k=3$ grades by applying the time-shifted quorum technique \emph{twice}. The first application happens during time~$[2\Delta, 4\Delta]$, and ensures the Graded Delivery property between grades~$0$ and $1$, exactly in the same way as in our Graded Agreement with $k=2$ grades. This application is \emph{nested} inside the second one, which ensures the Graded Delivery property between grades~$1$ and~$2$, and happens during time~$[\Delta, 5\Delta]$. Overall, the relevant inclusions are $V^{\Delta}_\Lambda \cap V^{5\Delta}_\Lambda \subseteq V^{2\Delta}_\Lambda \cap V^{4\Delta}_\Lambda \subseteq V^{3\Delta}_\Lambda$, for any $\Lambda$, and $S^{3\Delta} \subseteq S^{4\Delta} \subseteq S^{5\Delta}$ (each set can belong to a different validator). The protocol is {given in Figure~\ref{fig:ga-three-values}}. 
This protocol lasts $5\Delta$ time instead of $3\Delta$ and, similarly to the previous GA, it requires that the number of Byzantine validators is less than half the number of active validators. Therefore, it in particular works in the $(5\Delta, 0,\frac{1}{2})$-sleepy model. 

\begin{figure}[h!]
    \centering

    \begin{boxedminipage}[t]{\columnwidth}\small

        \textbf{Upon input} $\Lambda$, a validator~$v_i$ runs the following algorithm whenever awake. All validators awake at time~$3\Delta$ participate in the output phase for grade~$0$. A validator participates in the output phase for grade~$1$ at time~$4\Delta$ if {and only if} it was awake also at time~$2\Delta$. A validator participates in the output phase for grade~$2$ at time~$5\Delta$ if {and only if} it was awake also at time~$\Delta$. At any time, honest validators forward any message received. Up to two different \logmsgname{} messages per sender are forwarded upon reception.

\begin{enumerate}
    \item \textbf{Input phase, $(t = 0)$:} Broadcast $\sign{\logmsg{\Lambda}}{i}$.

    \item \textbf{$(t = \Delta)$:} Store $V^\Delta$.

    \item \textbf{$(t = 2\Delta)$:} Store $V^{2\Delta}$.

    \item \textbf{Output phase for grade $0$, $(t = 3\Delta)$:} 
    If $|V^{3\Delta}_{\Lambda}| > |S^{3\Delta}|/2$, output $(\Lambda, 0)$.

    \item \textbf{Output phase for grade $1$, $(t = 4\Delta)$:} 
    If awake at time $2\Delta$: If $|V^{2\Delta}_{\Lambda} \cap V^{4\Delta}_{\Lambda}| > |S^{4\Delta}|/2$, output $(\Lambda, 1)$.

    \item \textbf{Output phase for grade $2$, $(t = 5\Delta)$:} 
    If awake at time $\Delta$: If $|V^\Delta_{\Lambda} \cap V^{5\Delta}_{\Lambda}| > |S^{5\Delta}|/2$, output $(\Lambda, 2)$.
\end{enumerate}

\end{boxedminipage}
    \caption{Graded Agreement with grades $k=3$ grades -- protocol for validator~$v_i$.}
\label{fig:ga-three-values}
\end{figure}

\subsection{TOB-SVD}
\label{sec:our-tob}

We conclude this section with our main result, TOB-SVD, a dynamically available Total-Order Broadcast protocol. Byzantine Total-Order Broadcast protocols can be built using Graded Agreement primitives, with some protocols invoking Graded Agreement multiple times for a single decision~\cite{mr, mmr, mmr2}. At its core, dynamically available TOB aims to guarantee the total order delivery of messages among honest validators, even in the face of challenges like dynamic participation. However, existing TOB protocols~\cite{mr, mmr,mmr2,gl} have grappled with challenges related to latency~\cite{mr}, resilience~\cite{mmr}, and scalability~\cite{gl}. Multiple invocations of GA for a single decision, {meaning multiple voting phases per new block,} might exacerbate scalability and latency issues, especially when considering implementation in blockchain networks with hundreds of thousands of validators weighing in on every decision. 

Building upon these insights, our work introduces a GA-based dynamically available TOB protocol -- TOB-SVD -- which, akin to 1/4MMR~\cite{mmr}, {in the best case} necessitates only a single GA invocation for {new block}, but enhances adversarial resiliency to 1/2.

Our protocol, which works in the $(5\Delta, 2\Delta, \frac{1}{2})$-sleepy model, proceeds over a series of views, each spanning a duration of $4\Delta$. Every view~$v$ initiates a Graded Agreement~$GA_v$ with grades~$0$, $1$, and $2$, which extends and overlaps with the following~$GA_{v+1}$ during view~$v+1$. This structure implies that a single view does not encapsulate a full cycle of a Graded Agreement; instead, a~$GA_v$ initiated in view~$v$ concludes its operations only in the succeeding view~$v+1$.

\begin{figure*}[htbp]
\centering
\begin{tikzpicture}[scale=0.8, every node/.style={scale=0.8}] 

\def \l{4.5}
\pgfmathsetmacro\r{5*\l/4}
\def \f{\scriptsize}

\draw[thick, <->] (0, -4.5) -- (\l/4, -4.5);
\node[above] at (\l/8, -4.5) {$\Delta$};

\draw[thick, opacity=0.4] (0,-3.5) -- (\l,-3.5);
\draw[thick] (\l,-3.5) -- (2*\l,-3.5);
\draw[thick, opacity=0.4] (2*\l,-3.5) -- (3*\l,-3.5);

\draw[thick, opacity=0.4] (0,-3) -- (0,-4);
\draw[thick] (\l,-3) -- (\l,-4);
\draw[thick] (2*\l,-3) -- (2*\l,-4);
\draw[thick, opacity=0.4] (3*\l,-3) -- (3*\l,-4);

\node[below, opacity=0.4] at (\l/2,-3.75) {View~$v-1$};
\node[below] at (\l+\l/2,-3.75) {View~$v$};
\node[below, opacity=0.4] at (2*\l + \l/2,-3.75) {View~$v+1$};

\draw[dotted] (\l+ \l/4,-3.25) -- (\l + \l/4,-3.75);
\draw[dotted] (\l + 2*\l/4,-3.25) -- (\l + 2*\l/4,-3.75);
\draw[dotted] (\l+ 3*\l/4,-3.25) -- (\l + 3*\l/4,-3.75);
\draw[dotted] (\l + 4*\l/4,-3.25) -- (\l + 4*\l/4,-3.75);

\draw[dotted] (0,-3.25) -- (0,-3.75);
\draw[dotted] (\l/4,-3.25) -- (\l/4,-3.75);
\draw[dotted] (2*\l/4,-3.25) -- (2*\l/4,-3.75);
\draw[dotted] (3*\l/4,-3.25) -- (3*\l/4,-3.75);
\draw[dotted] (4*\l/4,-3.25) -- (4*\l/4,-3.75);

\draw[dotted] (2*\l+ \l/4,-3.25) -- (2*\l + \l/4,-3.75);
\draw[dotted] (2*\l + 2*\l/4,-3.25) -- (2*\l + 2*\l/4,-3.75);
\draw[dotted] (2*\l+ 3*\l/4,-3.25) -- (2*\l + 3*\l/4,-3.75);
\draw[dotted] (2*\l + 4*\l/4,-3.25) -- (2*\l + 4*\l/4,-3.75);

\node[anchor=base] at (\l,-2.8) {Propose};
\node[anchor=base] at (\l + \l/4,-2.8) {Vote};
\node[anchor=base] at (\l + 2*\l/4,-2.8) {Decide};

\node[anchor=base, opacity=0.4] at (2*\l,-2.8) {Propose};
\node[anchor=base, opacity=0.4] at (2*\l + \l/4,-2.8) {Vote};
\node[anchor=base, opacity=0.4] at (2*\l + 2*\l/4,-2.8) {Decide};

\node[anchor=base, opacity=0.4] at (0,-2.8) {Propose};
\node[anchor=base, opacity=0.4] at (0 + \l/4,-2.8) {Vote};
\node[anchor=base, opacity=0.4] at (0 + 2*\l/4,-2.8) {Decide};

\draw[thick, color=orange] (\l/4,-1.5) -- (\l + 2*\l/4,-1.5);

\draw[thick, color=orange] (\l/4,-1.25) -- (\l/4,-1.75);
\draw[thick, color=orange] (\l + 2*\l/4,-1.25) -- (\l + 2*\l/4,-1.75);

\node[below, color=orange] at (3.5*\l/4,-1.5) {$GA_{v-1}$};

\draw[dotted, color=orange, thick] (2*\l/4,-1.25) -- (2*\l/4,-1.75);
\draw[dotted, color=orange, thick] (3*\l/4,-1.25) -- (3*\l/4,-1.75);
\draw[dotted, color=orange, thick] (4*\l/4,-1.25) -- (4*\l/4,-1.75);
\draw[dotted, color=orange, thick] (5*\l/4,-1.25) -- (5*\l/4,-1.75);

\node[above, color=orange, font=\f] at (\l/4, -1.25) {Input};
\node[above, color=orange, font=\f] at (\l/4 + 3*\r/5, -1.25) {Out.~$0$};
\node[above, color=orange, font=\f] at (\l/4 + 4*\r/5, -1.25) {Out.~$1$};
\node[above, color=orange, font=\f] at (\l/4 + 5*\r/5, -1.25) {Out.~$2$};

\draw[->, thick, color=orange] (\l,-2) -- (\l,-2.5);
\draw[->, thick, color=orange] (\l + \l/4,-2) -- (\l + \l/4,-2.5);
\draw[->, thick, color=orange] (\l + 2*\l/4,-2) -- (\l + 2*\l/4,-2.5);

\draw[->, thick, color=orange] (\l + \l/4,-3) -- (\l + \l/4,-4.8);

\draw[->, thick, color=blue] (2*\l,-4.8) -- (2*\l,-4.3);
\draw[->, thick, color=blue] (2*\l + \l/4,-4.8) -- (2*\l + \l/4,-4.3);
\draw[->, thick, color=blue] (2*\l + 2*\l/4,-4.8) -- (2*\l + 2*\l/4,-4.3);

\draw[thick, color=blue] (\l + \l/4,-5.5) -- (2*\l+2*\l/4,-5.5);

\draw[thick, color=blue] (\l + \l/4,-5.25) -- (\l + \l/4,-5.75);
\draw[thick, color=blue] (2*\l + 2*\l/4,-5.25) -- (2*\l + 2*\l/4,-5.75);

\node[below, color=blue] at (\l + 3.5*\l/4,-5.5) {$GA_v$};

\draw[dotted, color=blue, thick] (\l + 2*\l/4,-5.25) -- (\l + 2*\l/4,-5.75);
\draw[dotted, color=blue, thick] (\l + 3*\l/4,-5.25) -- (\l + 3*\l/4,-5.75);
\draw[dotted, color=blue, thick] (\l + 4*\l/4,-5.25) -- (\l + 4*\l/4,-5.75);
\draw[dotted, color=blue, thick] (\l + 5*\l/4,-5.25) -- (\l + 5*\l/4,-5.75);

\node[above, color=blue, font=\f] at (\l + \l/4, -5.25) {Input};
\node[above, color=blue, font=\f] at (\l + \l/4 + 3*\r/5, -5.25) {Out.~$0$};
\node[above, color=blue, font=\f] at (\l + \l/4 + 4*\r/5, -5.25) {Out.~$1$};
\node[above, color=blue, font=\f] at (\l + \l/4 + 5*\r/5, -5.25) {Out.~$2$};

\end{tikzpicture}
\caption{In the middle, views~$v-1$, $v,$ and $v+1$ of our Total-Order Broadcast protocol, each with its three phases. At the top and bottom, respectively, $GA_{v-1}$ and $GA_v$ --- the Graded Agreement instances that are run as part of the TOB protocol. Arrows indicate that outputs of a GA are used by a parallel phase of the TOB and/or the next GA: outputs of grade~$0$ of $GA_{v-1}$ are extended by proposals of view~$v$, outputs of grade~$1$ of $GA_{v-1}$ are extended by votes in the TOB, which \emph{exactly correspond} to the inputs to $GA_v$, while outputs of grade~$2$ of $GA_{v-1}$ are decided in view~$v$.}
\label{fig:protocol-diagram}
\end{figure*}

Specifically, the protocol, which is presented in Figure~\ref{fig:tob1}, proceeds in views of $4\Delta$ time each. We let $t_v = 4\Delta v$ be the beginning of view~$v$. To each view $v$ corresponds a Graded Agreement $GA_v$, which runs in the time interval $[t_v + \Delta, t_v + 6\Delta] = [t_v + \Delta, t_{v+1} + 2\Delta]$, i.e., $GA_v$ takes up some of view $v+1$ as well. Moreover, the GA invocations are not perfectly sequential, as $GA_v$ and $GA_{v+1}$ overlap during time $[t_{v+1} + \Delta, t_{v+1} + 2\Delta]$ (Figure~\ref{fig:protocol-diagram}).

At the beginning of each view~$v$ there is a \emph{proposal phase} corresponding to the output phase for grade~$0$ of~$GA_{v-1}$. The leader of view~$v$ proposes a log~$\Lambda$ extending its grade~$0$ output (a \emph{candidate}), if it has one. 

Next, at time $t_v + \Delta$, comes a \emph{voting phase}, corresponding both to the input phase of $GA_v$ and to the output phase for grade~$1$ of $GA_{v-1}$. Validators which have a grade~$1$ output treat it as a \emph{lock}, and to preserve Safety they only input to~$GA_v$ either the lock itself or a proposal which extends it. 

This is followed by a \emph{decision phase} at time~$t_v + 2\Delta$, where logs proposed (at best) \emph{in the previous view} ($v-1$) can be decided. In particular, this phase corresponds to the output phase for grade~$2$ of~$GA_{v-1}$, and such grade~$2$ outputs are decided. This is safe, because Graded Delivery of Graded Agreement ensures that all honest validators (among those participating in the output phase for grade~$1$) output with grade~$1$ a decided log, thus locking on it and therefore inputting to $GA_v$ a log extending it. During the decision phase, validators also update~$V$ for $GA_v$ with all \logmsgname{} messages received until that point in time for~$GA_v$ {from non-equivocating validators}, as part of the ongoing $GA_v$. 

Finally, at time $t_v + 3\Delta$, only the action required by~$GA_v$ is performed, i.e., to store in~$V^{2\Delta}$ all \logmsgname{} messages received until that point from non-equivocating validators.

Whenever an action requires a $GA$ output which is not computed, i.e., the validator chooses not to participate in the output phase due to not having been previously awake when required, the action is skipped. In particular, no decision is taken at time $t_v + 2\Delta$ and no \logmsgname{} message is broadcast at time $t_v + \Delta$ when the required outputs are not available.

\begin{figure}[ht!]
    \centering

    \begin{boxedminipage}[t]{\columnwidth}\small

    Outputs of $GA_{-1}$, which are used in view $v=0$, are all formally defined to be the log containing only the Genesis {log}. In each view $v\ge 0$, awake validators participate in the GA instances that are ongoing, and in addition behave as specified here \emph{whenever they have the required GA outputs to do so}. Validators do not perform actions which require outputs they do not have.

\begin{enumerate}
    \item \textbf{Propose ($t=t_v$):} 
    Output phase for grade $0$ of $GA_{v-1}$. Propose $\Lambda'$ extending $\Lambda$, the highest log output with grade $0$ by $GA_{v-1}$ (\emph{candidate}), accompanied by {the} VRF value for view $v$.

    \item \textbf{Vote ($t=t_v + \Delta$):} 
    Output phase for grade $1$ of $GA_{v-1}$. $GA_v$ starts. Let $L_{v-1}$ be the highest log output with grade $1$ by $GA_{v-1}$ (\emph{lock}). After discarding equivocating proposals, input to $GA_v$ the proposal with the highest VRF value extending $L_{v-1}$, or $L_{v-1}$ if no such proposal exists.

    \item \textbf{Decide ($t=t_v + 2\Delta$):} 
    Output phase for grade $2$ of $GA_{v-1}$, which ends. Decide the highest log $\Lambda$ output with grade $2$ by $GA_{v-1}$.

    \item \textbf{($t=t_v + 3\Delta$):} 
    Do nothing (other than what is required by the ongoing $GA_v$).
\end{enumerate}

\end{boxedminipage}
    \caption{TOB-SVD: Total-Order Broadcast protocol with one {phase} of voting {per new block} -- protocol for validator~$v_i$.}
\label{fig:tob1}
\end{figure}

\noindent\textbf{On the number of grades.} Our protocol requires $k=3$ grades to reach a decision with a single instance of the GA. Specifically, grade 2 enables an instant decision, which would otherwise necessitate two instances of GA, each with $k=2$ grades. For comparison, in constructions like 1/3MMR, validators first input their logs into the first GA. Afterward, they input the longest log output with grade 1 from the first GA into the second GA. The decision is then made from the grade 1 output of the second GA, which corresponds to the grade 2 output in our GA with $k=3$ grades.

It is important to note that the decision phase of our TOB protocol is solely focused on making decisions and nothing else. By omitting the decision phase, the protocol can be simplified, still maintaining one GA per decision. However, this modification would shift the protocol from ensuring \emph{deterministic} Safety to \emph{probabilistic} Safety. As a result, it would align more closely with a different line of research focused on probabilistically safe, dynamically available protocols, such as those presented in~\cite{goldfish, rlmd}. 
However, exploring probabilistically safe TOB protocols falls outside the scope of this work.

\section{Analysis}
\label{sec:analysis}
{In this section, we analyze the protocols in Section~\ref{sec:tob-svd}.}

\subsection{Graded Agreement with $k=2$ grades}
\label{sec:2valGA-analysis}
{We begin by analyzing the Graded Agreement protocol with $k = 2$ grades.}
Recall that this protocol works in the $(3\Delta, 0,\frac{1}{2})$-sleepy model and that honest validators which are awake during the output phase for grade 1, at time~$t=3\Delta$, only participate in it \emph{if they were also awake at time~$t=\Delta$}.

\begin{theorem}
\label{thm:correctness-ga1}
    The protocol in Figure~\ref{fig:ga-two-values} implements Graded Agreement with $k=2$ grades.
\end{theorem}

\begin{proof}
    
For the \emph{Consistency} property we want to show that no two honest validators~$v_i$ and~$v_j$ output conflicting logs~$\Lambda$ and~$\Lambda'$ with grade~$1$. Without loss of generality, say that $|V^{\Delta, j}_{\Lambda'} \cap V^{3\Delta, j}_{\Lambda'}| \le |V^{\Delta, i}_{\Lambda} \cap V^{3\Delta, i}_{\Lambda}|$. We first show that the sets of senders in $V^{\Delta, j}_{\Lambda'} \cap V^{3\Delta, j}_{\Lambda'}$ and $V^{\Delta, i}_{\Lambda} \cap V^{3\Delta, i}_{\Lambda}$ are disjoint, and moreover that they are both contained in $S^{3\Delta, j}$. 

Let $m = (\Lambda'', v_k)$ for some $\Lambda'' \succeq \Lambda $ and some validator~$v_k$, and say $m \in V^{\Delta, i} \cap V^{3\Delta, i}$. Since $m \in V^{\Delta, i}$, validator~$v_i$ forwards {message $\sign{\logmsg{\Lambda''}}{k}$} at time $\Delta$ and~$v_j$ receives {it} by time~$2\Delta$. Therefore, {either $m \in V^{3\Delta, j}$ or  $v_k \notin V^{3\Delta, j}$}, because receiving any {log other than $\Lambda''$ {from} $v_k$} would lead to~$v_j$ having equivocation evidence for~$v_k$. Since $m$ is for $\Lambda''$, extending $\Lambda$ and thus conflicting with $\Lambda'$, $V^{3\Delta, j}_{\Lambda'}$ does not include any log from~$v_k$. 
Hence, the senders of $V^{\Delta, j}_{\Lambda'} \cap V^{3\Delta, j}_{\Lambda'}$ and $V^{\Delta, i}_{\Lambda} \cap V^{3\Delta, i}_{\Lambda}$ are disjoint. Moreover, the senders of logs in $V^{\Delta, j}_{\Lambda'} \cap V^{3\Delta, j}_{\Lambda'}$ are by definition contained in $S^{3\Delta, j}$. Finally, the senders of logs {in} $V^{\Delta, i}_{\Lambda}$ are also all contained in $S^{3\Delta, j}$, because~$v_i$ forwards them at time $\Delta$ and, by time $2\Delta$, ~$v_j$ accepts either them or equivocation evidence for their senders. 

It then follows that $|S^{3\Delta, j}| \ge |V^{\Delta, i}_{\Lambda} \cap V^{3\Delta, i}_{\Lambda}| + |V^{\Delta, j}_{\Lambda'} \cap V^{3\Delta, j}_{\Lambda'}|$. Then, $|V^{\Delta, j}_{\Lambda'} \cap V^{3\Delta, j}_{\Lambda'}| \le |V^{\Delta, i}_{\Lambda} \cap V^{3\Delta, i}_{\Lambda}|$ implies $|S^{3\Delta, j}| \ge 2|V^{\Delta, j}_{\Lambda'} \cap V^{3\Delta, j}_{\Lambda'}|$, so $v_j$ does not output $(\Lambda', 1)$. 

For the \emph{Graded Delivery} property we show that an honest validator~$v_i$ outputting $(\Lambda, 1)$ implies that any honest validator~$v_j$ participating in the output phase for grade~$0$ outputs $(\Lambda, 0)$. At time~$\Delta$, validator~$v_i$ forwards all messages in $V^{\Delta, i}$. At time~$2\Delta$, validator~$v_j$ has equivocation evidence for the sender of any log in $V^{\Delta, i} \setminus V^{2\Delta, j}$, since otherwise such a log would also be contained in $V^{2\Delta, j}$. This equivocation evidence is forwarded by~$v_j$ and received by~$v_i$ by time~$3\Delta$, so the senders of logs in $V^{\Delta, i} \setminus V^{2\Delta, j}$ are all considered as equivocators by~$v_i$ then. This implies that $V^{\Delta, i} \setminus V^{2\Delta, j}$ and $V^{3\Delta, i}$ are disjoint, and thus so are $V^{\Delta, i} \cap V^{3\Delta, i}$ and $V^{\Delta, i} \setminus V^{2\Delta, j}$. Therefore, $V^{\Delta, i} \cap V^{3\Delta, i} \subseteq V^{2\Delta, j}$. Moreover, $S^{2\Delta, j} \subseteq S^{3\Delta, i}$, given that~$v_i$ forwards at least a \logmsgname{} message for each sender in $S^{2\Delta, j}$ by time~$2\Delta$, and by time~$3\Delta$ validator~$v_i$ receives either the forwarded \logmsgname{} messages or equivocation evidence for each sender. Validator~$v_i$ outputs~$(\Lambda, 1)$ if $|V^{\Delta, i}_{{\Lambda}} \cap V^{3\Delta, i}_{{\Lambda}}| > |S^{3\Delta, i}|/2$, in which case we also have $|V^{2\Delta, j}_{{\Lambda}}| \ge |V^{\Delta, i}_{{\Lambda}} \cap V^{3\Delta, i}_{{\Lambda}}| > |S^{3\Delta, i}|/2 \ge |S^{2\Delta, j}|/2$. Thus, if validator~$v_i$ outputs $(\Lambda, 1)$, validator~$v_j$ outputs $(\Lambda, 0)$.

For the \emph{Validity} property, consider an honest validator $v_i \in H_{\Delta} \cap H_{3\Delta}$ participating in the output phase for grade~1, i.e., $v_i$ awake at $t=\Delta$ and $t=3\Delta$. Suppose that all the validators $H_0 \setminus {B_{3\Delta}}$ (initially awake and honest throughout the protocol) send logs that extends $\Lambda$. Since such validators are honest throughout the protocol, they never equivocate, so $V^{\Delta, i}_{\Lambda}$ and $V^{3\Delta, i}_{\Lambda}$ contain all of these logs. 
Therefore, $|V^{\Delta, i}_{\Lambda} \cap V^{3\Delta, i}_{\Lambda}| \ge |H_0 \setminus {B_{3\Delta}}| = |H_0 \cup {B_{3\Delta}}| - |{B_{3\Delta}}| $.
Note that $H_0 \cup B_{3\Delta}$ contains all validators which might ever send a log during the whole protocol.
Then, due to Condition~\eqref{eq:sleepy} and $\rho = \frac{1}{2}$,   $|V^{\Delta, i}_{\Lambda} \cap V^{3\Delta, i}_{\Lambda}| > |S^{3\Delta, i}|/2$.
It follows that $v_i$ outputs $(\Lambda, 1)$. The argument for grade $0$ is nearly identical, and we omit it.

\emph{Integrity} follows from $\rho = \frac{1}{2}$ and the fact that outputting $\Lambda$ with any grade requires a majority of unique log senders to have sent a log that extends $\Lambda$ \emph{and to not have equivocated}. 

Consider an honest validator~$v_i$ which outputs $\Lambda$ with grade~$0$, after observing $|V^{2\Delta, i}_\Lambda| > |S^{2\Delta, i}|/2$ at time $2\Delta$, though no validator in~$H_0$ has sent a log $\Lambda' \succeq \Lambda$ {at time~$t=0$}. 
So, $V^{2\Delta, i}_{\Lambda}$ does not contain any log received from validators in $H_0$.
However, both the validators in $V^{2\Delta, i}_{\Lambda}$ and the ones in $H_0$ are counted as senders in $S^{2\Delta, i}$, and therefore, we have $|S^{2\Delta, i}| \ge |V^{2\Delta, i}_{\Lambda}| + |H_0|$.
Note that $|H_0 \cup B_{3\Delta}| \geq |S^{2\Delta, i}|$. 
Then, due to Condition~\eqref{eq:sleepy} and $\rho = \frac{1}{2}$,  $|H_0| \geq |H_0 \cup B_{3\Delta}| - |B_{3\Delta}| > |H_0 \cup B_{3\Delta}|/2 \ge |S^{2\Delta, i}|/2$.
Then, $|V^{2\Delta, i}_{\Lambda}| > |S^{2\Delta, i}|/2$ and $|S^{2\Delta, i}| \ge |V^{2\Delta, i}_{\Lambda}| + |H_0|$ together imply $|S^{2\Delta, i}| > |S^{2\Delta, i}|/2 + |H_0| > |S^{2\Delta, i}|/2 + |S^{2\Delta, i}|/2$, i.e., $|S^{2\Delta, i}| > |S^{2\Delta, i}|$, a contradiction. 
The argument for grade~$1$ is almost identical, and we omit it.

Finally, similarly to Integrity, \emph{Uniqueness} follows from the fact that outputting $\Lambda$ with any grade requires a majority of unique log senders to have sent a log that extends $\Lambda$, \emph{without counting logs received from equivocators}. This ensures that the sets of senders of logs which a validator counts in support of conflicting logs do not intersect, and so that a validator cannot see a majority for two conflicting logs.  We only go through the argument for grade $0$, because Uniqueness for grade $1$ is already implied by Consistency. Note first that there is a natural injection of $V^{2\Delta, i}_{\Lambda}$ into $S^{2\Delta, i}$, since at most one log per validator is considered in the former, due to removing logs from equivocators. Moreover, note that $V^{2\Delta, i}_{\Lambda'}$ and $V^{2\Delta, i}_{\Lambda}$ are disjoint for conflicting {$\Lambda$ and $\Lambda'$}, so $|V^{2\Delta, i}_{\Lambda'}| + |V^{2\Delta, i}_{\Lambda}| \leq |S^{2\Delta, i}|$. Therefore, $|V^{2\Delta, i}_{\Lambda}| > |S^{2\Delta, i}|/2$ implies $|V^{2\Delta, i}_{\Lambda'}| \leq |S^{2\Delta, i}|/2$ for any conflicting $\Lambda'$.
\end{proof}

\subsection{Graded Agreement with $k=3$ grades}
\label{sec:3valGA-analysis}
{We now analyze the Graded Agreement protocol with $k=3$ grades.}
Recall that this protocol works in the $(5\Delta, 0,\frac{1}{2})$-sleepy model. 

During the output phase for grade~$1$, honest validators that are awake at time $4\Delta$ only participate if they were also awake at time $2\Delta$. In the output phase for grade~$2$, validators awake at time $5\Delta$ will only participate if they were awake at time $\Delta$. This is because outputting a log with either grade~$1$ or~$2$ requires having \emph{previously} stored supporting logs for them, as per the time-shifted quorum technique.

\begin{theorem}
\label{thm:correctness-ga2}
    The protocol in Figure~\ref{fig:ga-three-values} implements Graded Agreement with $k=3$ grades.
\end{theorem}

\begin{proof}
The proof for all the properties of the Graded Agreement with $k=3$ grades are similar to those for the Graded Agreement with $k=2$ grades (Theorem~\ref{thm:correctness-ga1}). For this reason, we only discuss the \emph{Graded Delivery} property, which is where the key idea of the protocol is utilized, i.e., a nested application of the time-shifted quorum technique. Let $v_i$, $v_j$, and $v_k$ be three honest validators participating in the output phase for grade~$0$, grade~$1$, and grade~$2$, respectively. Time $[2\Delta, 4\Delta]$ in this GA functions exactly like time $[\Delta, 3\Delta]$ in the Graded Agreement with $k=2$ grades: at first, logs $V^{2\Delta}$ are stored for later use, then comes the output phase for grade~$0$, where a log $\Lambda$ is output with grade~$0$ if $|V^{3\Delta}_{\Lambda}| > |S^{3\Delta}|/2$, and finally comes the output phase for grade~$1$, where $\Lambda$ is output with grade~$1$ if $|V^{2\Delta}_{\Lambda} \cap V^{4\Delta}_{\Lambda}| > |S^{4\Delta}|/2$. In other words, for grade~$0$ and grade~$1$ we have a first application of the time-shifted quorum technique. For the same reasons as in the Graded Agreement with $k=2$ grades, we then have that $|V^{2\Delta, j}_\Lambda \cap V^{4\Delta, j}_\Lambda| \le |V^{3\Delta,i}_\Lambda|$ and $|S^{3\Delta, i}| \le |S^{4\Delta,j}|$, which guarantees the Graded Delivery property from grade~$1$ to grade~$0$. This application of the time-shifted quorum technique is \emph{nested inside another such application}, which guarantees the Graded Delivery property from grade~$2$ to grade~$1$. Firstly, the participation level for grade~$2$ outputs, i.e., $S^{5\Delta}$, is determined time $\Delta$ \emph{after} that for grade~$1$ outputs, $S^{4\Delta}$, which ensures the correct inclusion, i.e., $S^{4\Delta, j} \subseteq S^{5\Delta, k}$.  Conversely, $V^\Delta_\Lambda$, the initial supporting logs for grade~$2$ outputs, are determined time $\Delta$ \emph{before} $V^{2\Delta}_\Lambda$, those for grade~$1$ outputs. Finally, the sets of equivocating senders, whose votes are discarded from the supporting logs, are determined in the same order as the participation level, to ensure that any sender which is considered an equivocator in output phase of grade~$1$ is also considered an equivocator in the output phase for grade~$2$. Together, these last two points guarantee that $V^{\Delta, k}_{\Lambda} \cap V^{5\Delta, k}_{\Lambda} \subseteq V^{2\Delta, j}_{\Lambda} \cap V^{4\Delta, j}_{\Lambda}$, so that $|V^{\Delta, k}_{\Lambda} \cap V^{5\Delta, k}_{\Lambda}| > |S^{5\Delta, k}|/2$ implies $|V^{2\Delta, j}_{\Lambda} \cap V^{4\Delta, j}_{\Lambda}| \ge |V^{\Delta, k}_{\Lambda} \cap V^{5\Delta, k}_{\Lambda}| > |S^{5\Delta, k}|/2 \ge |S^{4\Delta, j}|/2$, i.e., validator~$v_k$ outputting $(\Lambda, 2)$ implies validator~$v_j$ outputting $(\Lambda, 1)$.
\end{proof}

\subsection{TOB-SVD}
\label{sec:tob-svd-analysis}

{We conclude by analyzing our main result, TOB-SVD.}

\begin{theorem}
\label{thm:correctness-tob1}
    The protocol in Figure~\ref{fig:tob1} implements Total-Order Broadcast.
\end{theorem}

In order to prove this theorem, we first ensure that the properties of our Graded Agreement primitive (Figure~\ref{fig:ga-three-values}) hold when employing it within the TOB-SVD protocol (Figure~\ref{fig:tob1}). 

Observe that, due to the assumed stabilization period of~$2\Delta$, our Total-Order Broadcast works in the $(5\Delta, 2\Delta, \frac{1}{2})$-sleepy model, i.e., when $|B_{t + 5\Delta}| <  |H_{t - 2\Delta, t} \cup B_{t + 5\Delta}| / 2$  holds for every time~$t \ge 0$. On the other hand, the Graded Agreement protocol itself  (Figure~\ref{fig:ga-three-values}) works in the $(5\Delta, 0, \frac{1}{2})$-sleepy model, requiring $|B_{t + 5\Delta}| <  |H_{t, t} \cup B_{t + 5\Delta}| / 2$, without any stability assumption ($T_s = 0$). When invoking it as part of our Total-Order Broadcast, the only change is in the input phase, since honest validators use the outputs of $GA_{v-1}$ to determine the inputs to~$GA_{v}$. In particular, validators in $H_{t_v + \Delta}$ input something to $GA_v$ if and only if they have a \emph{lock} $L_{v-1}$, such that they output $L_{v-1}$ with grade~$1$ in $GA_{v-1}$. For that to be the case, they must have participated in the output phase for grade~$1$ of $GA_{v-1}$, which requires them to also have been awake at time $t_{v-1} + 3\Delta = t_v - \Delta$. The additional stability requirement of $T_s = 2\Delta$ in the Total-Order Broadcast model takes care of this, ensuring that we only consider those honest validators which are allowed to participate in the input phase of a~$GA_v$, i.e., validators in $H_{t-2\Delta, t}$, for $t = t_v + \Delta$.

\begin{lemma}
\label{lem:tob1-induction}
If all honest validators participating in the output phase for grade~$1$ of $GA_{v-1}$ output~$(\Lambda, 1)$
then, for any view $v' \ge v-1$, all honest validators participating in the output phase for grade~$1$ of $GA_{v'}$ output~$(\Lambda, 1)$.

\end{lemma}

\begin{proof} 
    By the Uniqueness property of the Graded Agreement, any honest validator~$v_i$ that outputs {a log extending}~$\Lambda$ with grade~$1$ in~$GA_{v-1}$ does not output any log conflicting with~$\Lambda$ with grade~$1$. This means that the lock $L_{v-1}$ of~$v_i$ extends $\Lambda$. Therefore, every honest validator inputs to $GA_v$ a log extending~{$\Lambda$}. Since the honest validators participating in the output phase for grade~$1$ of $GA_{v-1}$ exactly correspond to those that input something to $GA_v$, we can apply the Validity property of $GA_v$ and conclude that all honest validators that participate in the output phase for grade~$1$ of~$GA_v$ output $(\Lambda, 1)$. By induction, this then holds for all views~$v' \ge v-1$.
\end{proof}

\begin{theorem}[Safety]
    The protocol implemented in Figure~\ref{fig:tob1} satisfies Safety.
\end{theorem}

\begin{proof}
    Suppose an honest validator~$v_i$ decides log $\Lambda$ at time $t_v + 2\Delta$ by outputting $(\Lambda, 2)$ in $GA_{v-1}$. By the {Graded Delivery} property, any honest validator participating in the output phase for grade~$1$ of $GA_{v-1}$ outputs~$(\Lambda, 1)$. By Lemma~\ref{lem:tob1-induction}, for any $v' \ge v-1$, honest validators participating in the output phase for grade~$1$ of $GA_{v'}$ output $(\Lambda, 1)$. Now, suppose that another honest validator~$v_j$ decides a conflicting~$\Lambda'$ and, without loss of generality, let us assume that~$v_j$ does so during view $v'' \ge v$. Again, by the {Graded Delivery} property, every honest validator which participates in the output phase of $GA_{v''-1}$ outputs $(\Lambda', 1)$. Since $v''-1 \ge v-1$ we have shown that any such validator also outputs $(\Lambda, 1)$, contradicting the Uniqueness property of Graded Agreement.
\end{proof}

\begin{lemma}
\label{lem:good-leader-probability}
    Any view has a good leader with probability greater than $\frac{1}{2}$.
\end{lemma}

\begin{proof}
    Observe that, due to Condition~\eqref{eq:sleepy}, $\rho = \frac{1}{2}$ and $B_{t_v} \subseteq B_{t_v + 5\Delta}$, $|H_{t_v} \setminus {B_{t_v + \Delta}}| \geq |H_{t_v} \setminus {B_{t_v + 5\Delta}}| = |H_{t_v} \cup {B_{t_v + 5\Delta}}| - |{B_{t_v + 5\Delta}}|  > |H_{t_v} \cup {B_{t_v + 5\Delta}}|/2 \geq |H_{t_v} \cup {B_{t_v + \Delta}}|/2$.
    View $v$ has a good leader whenever a validator in $H_{t_v} \setminus {B_{t_v + \Delta}}$
    has the highest VRF value for view $v$ out all validators in $H_{t_v} \cup B_{t_v + \Delta}$. 
    The adversary is mildly adaptive, so corruptions which happen by time $t_v + \Delta$ must have been scheduled by time $t_v$. In particular, the adversary has to determine $H_{t_v} \cap B_{t_v + \Delta}$ \emph{before} observing any of the VRF values of validators in $H_{t_v}$. Therefore, view $v$ has a good leader with probability $\frac{|H_{t_v} \setminus {B_{t_v + \Delta}}|}{|H_{t_v} \cup B_{t_v + \Delta}|} > \frac{1}{2}$.
\end{proof}

\begin{lemma}
\label{lem:tob1-leader-followed}
    If view $v$ has a good leader~$v_{\ell}$ and~$v_{\ell}$ proposes a log $\Lambda$, then all honest validators participating in the output phase for grade~$1$ of $GA_{v-1}$ input~$\Lambda$ to $GA_v$.
\end{lemma}

\begin{proof}
    Consider any such honest validator~$v_i$ and its lock $L_{v-1}$, which~$v_i$ outputs with grade~$1$ in $GA_{v-1}$.
    {Note that by the Validity property and the fact that any log is an extension of $\genesis$, any honest validator participating in the output phase for grade~$1$ outputs some log extending $\genesis$.}
    As the leader~$v_{\ell}$ is honest and awake at time~$t_v$, by the Graded Delivery {property} of Graded Agreement,~$v_{\ell}$ outputs~$(L_{v-1}, 0)$ in $GA_{v-1}$ , and, by the Uniqueness property, it does not output any conflicting log with grade~$0$. This means that the proposal $\Lambda$ made by leader~$v_{\ell}$ extends~$L_{v-1}$. The proposal is received by validator~$v_i$ by time~$t_v + \Delta$, and no other proposal from $v_{\ell}$ is received by~$v_i$ at that point, because the leader is still honest at time $t_v + \Delta$, since $v_{\ell} \not \in B_{t_v + \Delta}$ by definition of a good leader. Moreover, no other proposal received by~$v_i$ at this point has higher VRF value, since a good leader for view $v$ has the highest VRF value among all validators from which a proposal might have been received by ($H_{t_v} \cup B_{t_v + \Delta})$. Therefore, validator~$v_i$ inputs $\Lambda$ to $GA_{v}$.
\end{proof}

\begin{lemma}
\label{lem:tob1-reorg-resilience}
In the protocol in Figure~\ref{fig:tob1}, {if view $v$ has a good leader~$v_{\ell}$ and~$v_{\ell}$ proposes a log $\Lambda$, then any honest validator that after $t_{v+1} - 2\Delta$ is eventually awake for at least $8\Delta$ decides a log extending $\Lambda$.}
\end{lemma}

\begin{proof}
    Suppose view $v$ has a good leader, which proposes log $\Lambda$. Then, by Lemma~\ref{lem:tob1-leader-followed}, all honest validators which participate in the output phase for grade $1$ of $GA_{v-1}$ input $\Lambda$ to $GA_v$. By the Validity property of Graded Agreement, all validators which participate in the output phase for grade $1$ of $GA_v$ output $(\Lambda, 1)$. By Lemma~\ref{lem:tob1-induction}, this also holds in $GA_{v'}$ for all $v' \ge v$. Since all such validators output $(\Lambda, 1)$, they also input an extension of $\Lambda$ to $GA_{v'+1}$, for all $v' \ge v$.
    Again by the Validity property of Graded Agreement, any honest validator which participates in the output phase for grade~$2$ of one such {$GA_{v'}$}, i.e., any honest validators awake both at {$t_{v'+1} - 2\Delta$ and $t_{v'+1} + 2\Delta$}, decides {a log extending}~$\Lambda$.
    {Any honest validator that after $t_{v+1} - 2\Delta$ is eventually awake for $8\Delta$ is guaranteed to be awake at both $t_{v''} - 2\Delta$ and $t_{v''} + 2\Delta$ for some $v'' \geq v+1$ and therefore it decides a log extending~$\Lambda$.} 
\end{proof}

\begin{theorem}[Liveness]
    The protocol implemented in Figure~\ref{fig:tob1} satisfies Liveness.
\end{theorem}

\begin{proof}
    Take a valid transaction $\mathit{tx}$ {in the pool of transactions.}
    By Lemma~\ref{lem:good-leader-probability} and our assumption on transaction batching by honest validators, there exists a view $v$ with an honest leader which proposes a log $\Lambda$ that includes $\mathit{tx}$.
    As per Lemma~\ref{lem:tob1-reorg-resilience}, $\mathit{tx}$ is included in the log decided by any honest validator that after $t_{v+1} - 2\Delta$ is eventually awake for at least $8\Delta$.
\end{proof}

\section{Conclusions}
\label{sec:conclusions}

We introduce TOB-SVD, a novel Total-Order Broadcast protocol that supports dynamic participation and operates within (a variant of) the sleepy model. TOB-SVD tolerates up to 1/2 Byzantine validators and improves both expected latency and transaction latency compared to existing protocols with the same resilience, requiring only one voting phase in the best case and two in expectation. This makes it significantly more practical for large-scale networks like permissionless blockchains. The protocol is also notably simpler, at the cost of an added \(2\Delta\) stabilization time on top of the time required by other protocols to retrieve information upon waking. Therefore, in scenarios where a large amount of information must be retrieved, such as in blockchain systems, our protocol requires only a slightly stronger assumption on the overall stabilization time.

\bibliography{main}
\bibliographystyle{plainurl}

\end{document}